\documentclass[a4paper, 10pt, conference]{ieeeconf}

\IEEEoverridecommandlockouts

\usepackage{graphics}
\usepackage{epsfig}
\usepackage{times}
\usepackage{amsmath}
\usepackage{amssymb}
\usepackage{multirow}
\usepackage{array}
\usepackage{amsfonts}

\usepackage{amsthm}
\usepackage{todonotes}
\usepackage{comment}
\usepackage{cite}
\usepackage{booktabs}

\title{\LARGE \textbf{A Game-Theoretic Analysis of the Social Impact of Connected and Automated Vehicles}}

\author{Ioannis Vasileios Chremos, \textit{Student Member, IEEE}, Logan E. Beaver, \textit{Student Member, IEEE} \\ Andreas A. Malikopoulos, \textit{Senior Member, IEEE}
\thanks{This research was supported in part by ARPAE's NEXTCAR program under the award number DE-AR0000796 and by the Delaware Energy Institute (DEI).}
\thanks{The authors are with the Department of Mechanical Engineering, University of Delaware, Newark, DE 19716 USA (emails: \tt\small{ichremos@udel.edu};\tt\small{lebeaver@udel.edu};
    \tt\small{andreas@udel.edu}.)}%
}
\theoremstyle{definition}
\newtheorem{assumption}{Assumption}
\newtheorem{definition}{Definition}
\theoremstyle{plain}
\newtheorem{theorem}{Theorem}
\newtheorem{corollary}{Corollary}
\newtheorem{lemma}{Lemma}
\newtheorem{proposition}{Proposition}
\newtheorem{remark}{Remark}

\theoremstyle{definition}
\newtheorem{example}{Example}

\begin{document}

\maketitle
\thispagestyle{empty}
\pagestyle{empty}

\begin{abstract}

In this paper, we address the much-anticipated deployment of connected and automated vehicles (CAVs) in society by modeling and analyzing the social-mobility dilemma in a game-theoretic approach. We formulate this dilemma as a normal-form game of players making a binary decision: whether to travel with a CAV (CAV travel) or not (non-CAV travel) and by constructing an intuitive payoff function inspired by the socially beneficial outcomes of a mobility system consisting of CAVs. We show that the game is equivalent to the Prisoner's dilemma, which implies that the rational collective decision is the opposite of the socially optimum. We present two different solutions to tackle this phenomenon: one with a preference structure and the other with institutional arrangements. In the first approach, we implement a social mechanism that incentivizes players to non-CAV travel and derive a lower bound on the players that ensures an equilibrium of non-CAV travel. In the second approach, we investigate the possibility of players bargaining to create an institution that enforces non-CAV travel and show that as the number of players increases, the incentive ratio of non-CAV travel over CAV travel tends to zero. We conclude by showcasing the last result with a numerical study.

\end{abstract}

\section{Introduction}\label{sec:intro}

The reality of connected and automated vehicles (CAVs) is coming fast to realization \cite{Marletto2019}. Similarly, with other past technologies, CAVs promise to be an incoming disruptive innovation with vast technological, commercial, and regulatory dimensions. Recently, there has been a significant amount of work on the technological or social impact of CAVs (mostly focusing on congestion, emissions, energy consumption, and safety). CAVs will transform the transportation system of today and revolutionize mobility. On the other hand, one expected social consequence of CAVs is to reshape urban mobility in the sense of altered tendency-to-travel, and thus, highly increase demand in the transportation system. To elaborate on this point, evident from similar technological revolutions (e.g., elevators), human social tendencies and society's perspective have changed the way a technology is used and applied \cite{Bernard2014}. Thus, we can most certainly expect that the deployment of CAVs in society will have unexpected outcomes, in the form of rebound effects (e.g., increased overall vehicle miles traveled, decreased use of public transportation, higher demand for road usage, etc.) Although there have been numerous studies that provide qualitative analysis for the social impact \cite{Zmud2017,Taiebat2018}, our game-theoretical approach aims to provide a formal analysis of the human decision-making regarding the expected social-mobility dilemma of the future travelers.

One may ask ``Why do we use Game Theory to analyze such a problem?" It is the authors' belief that the emerging transportation systems - CAVs, shared mobility, electric vehicles - will be characterized by their socio-economic complexity: (1) improved productivity and energy efficiency, (2) widespread accessibility, and (3) drastic urban redesign and evolved urban culture. This characteristic can naturally be modeled and analyzed using notions from Mathematical Psychology and Game Theory. One of the main arguments in this paper is that the social interaction of humans and CAVs can be modeled as a ``social dilemma." That is, we are only concerned with the impact of the human decision before the vehicle's engine is even turned on. Informally, a social dilemma is any situation where there is a subtle yet unwanted discrepancy between individual and collective interest. It is for this reason why the authors of this paper argue that social dilemmas are the appropriate models to be looking at instead of, for example, congestion games. We want to emphasize that we are interested in the human choice of commute and not the selfish routing on a road network.



By considering a normal-form game of $n$ players, we acquire a significantly improved way to realistically model social dilemmas that occur in real-life, and most importantly, we obtain a multiplayer structure that reflects Garrett Hardin's ``Tragedy of the Commons." From its conception, the Tragedy of the Commons has been an important problem in economics and other fields as it describes a plethora of phenomena in which independent members of a society selfishly attempt to maximize their benefit of utilizing at least one common resource which is scarce. Thus, the individuals' selfishness leads to the collective degradation of society's well-being. Noteworthy, even though the decision-makers are selfish and their decisions aim to maximize personal gain, they end up depleting the resource with unavoidable repercussions and losses \cite{Dawes1973}. In our context, the common resource is the road infrastructure shared by all the travelers, and the utilization is whether to travel with a CAV or not. Intuitively, one can expect that if all travelers make the selfish decision to use a CAV for commuting, then congestion is unavoidable.

In the first decades of the $20$th century, Arthur Cecil Pigou argued that ``if a system's decision-makers take autonomous decisions, then the resulting collective outcome most probably will be inefficient." This key observation is evident in many different studies and analyses of in transportation problems.
Social dilemmas have been extensively studied for systems that exhibit overpopulation, resource depletion, gridlock, and pollution \cite{Dawes1980,Platt1973,Stern1992}, while the Prisoner's dilemma (PD) game has been used to model vehicle congestion in a transportation network \cite{Joshi2005,Whitelegg1997} where travelers decide to continue driving their vehicles in congested and polluted cities. One possible solution to the PD was studied in \cite{Okada1993} in which inspired from notions of classical arguments on the theory of social contract, the author investigated whether cooperation might emerge in a social dilemma game with institutional arrangements.
Although the social effect of selfish-mobility behavior in routing networks of regular and autonomous vehicles has been studied \cite{Mehr2018}, it seems that the problem of how CAVs will affect human tendency-to-travel and mobility frequency has not been adequately approached yet. On the other hand, analytical frameworks have been proposed to quantify and evaluate the impacts of CAVs from the technological perspective \cite{Jackeline2016a,Jackeline2016b}. Furthermore, coordination of CAVs at different traffic scenarios, e.g., intersections or vehicle-following, have been extensively evaluated in the literature \cite{Malikopoulos2018,Jackeline2017}. Recently, there has been research done on the rebound effects which might arise from the introduction of automation in a transportation system \cite{Singleton2018,Soteropoulos2019}. For a detailed analysis of the effects of CAVs technologies on travel demand, see \cite{Auld2017}. Recently, in the literature, it has been recognized that further research is required to identify and understand the potential impacts of emerging mobility \cite{Sarkar2016,Zhao2019}.


The contributions of this paper are:
\begin{enumerate}
    \item we provide a game-theoretic analysis of the conflict of interest and model the social-mobility dilemma as a social dilemma, and
    \item we apply two different in mindset mechanisms and approaches that attempt to prevent negative outcomes, e.g., similar to the Tragedy of the Commons.
\end{enumerate}
Several research efforts reported in the literature have focused on studying social behavior regarding semi-autonomous driving and the selfish social decision-making of choosing a route to commute in a transportation network \cite{Mehr2019a}. A key difference between our work and the frameworks already reported in the literature is that we focus on modeling the human decision-making of which mode of transportation to be used rather than modeling selfish routing. Our analysis will complement these efforts by providing a framework that attempts to integrate the human social behavior in a mobility system consisting of CAVs. Moreover, our work in this paper expands the much-needed discussion on understanding the social impact and implications of CAVs by providing insights on how human behavior might react to an emerging mobility system. More specifically, our most important contribution is to rigorously show that without a well-thought intervention via regulations or incentives, a society of selfish travelers will make the wrong collective decisions, and thus, we will end up with a catastrophically sub-optimal performance of the emerging mobility system.

The remaining of the paper proceeds as follows. In Section \ref{sec:math_prelims}, we provide an overview of Game Theory notions. In Section \ref{sec:game_formulation}, we present our formulation of the social decision-making regarding the CAVs as a normal-form game and show that it is equivalent to a PD game. In Section \ref{sec:first_approach}, we introduce and study a preference structure, and in Section \ref{sec:second_approach}, we apply a framework of institutions and provide a numerical study of the results. Finally, we offer some concluding remarks and discuss future work in Section \ref{sec:conclusions}.

\section{Mathematical Preliminaries}\label{sec:math_prelims}

In this section, we present a brief overview of important notions from non-cooperative Game Theory. First, we assume that the players of the game are \textit{rational}, in the sense that each player's objective is to maximize the expected value of her own payoff. In addition, we assume that the players are \textit{intelligent}, i.e., each player has full knowledge of the game and has the ability to make any inferences about the game that we, the designers, can make. In order to develop a rigorous framework that analyzes the social dilemma as a game, we need to formally define a few important notions of Game Theory that will prove instrumental in our analysis.

\begin{definition}\label{definition1}
    A finite normal-form game is a tuple $\mathcal{G} = \langle \mathcal{I}, \mathcal{S}, (u_i)_{i \in \mathcal{I}} \rangle$, where
        \begin{itemize}
            \item $\mathcal{I} = \{1, 2, \ldots, n\}$ is a finite set of $n$ players with $n \geq 2$;
            \item $\mathcal{S} = S_1 \times \cdots \times S_n$, where $S_i$ is a finite set of actions available to player $i \in \mathcal{I}$ with $s = (s_1, \ldots, s_n) \in \mathcal{S}$ being the action profile;
            \item $u = (u_1, \ldots, u_n)$, where $u_i : \mathcal{S} \to \mathbb{R}$, is a real-valued utility function for player $i \in \mathcal{I}$.
        \end{itemize}
\end{definition}

\begin{definition}
    Let $S_{i}$ be the strategy profile of player $i$, $s_i, s_i ' \in S_{i}$ be two strategies of player $i$, and $S_{- i}$ be the set of all strategy profiles of the remaining players. Then, $s_i$ strictly dominates $s_i '$ if, for all $s_{- i} \in S_{- i}$, we have $u_i(s_i, s_{- i}) > u_i(s_i ', s_{- i})$. Also, a strategy is strictly dominant if it (strictly) dominates any other strategy.
\end{definition}

\begin{definition}\label{definition2}
    A player $i$'s best response to the strategy profile $s_{- i} = (s_1, \ldots, s_{i - 1}, s_{i + 1}, \ldots, s_n)$ is the strategy $s ^ *_i \in S_i$ such that $u_i(s ^ *_i, s_{- i}) \geq u_i(s_i, s_{- i})$ for all $s_i \in S_i$. A strategy profile $s$ is a Nash equilibrium (NE) if, for each player $i$, $s_i$ is a best response to $s_{- i}$.
\end{definition}

Next, for completeness, we define the notion of Pareto domination. First, an ``outcome" of a game is any strategy profile $s \in \mathcal{S}$. Intuitively, an outcome that Pareto dominates some other outcome improves the utility of at least one player without reducing the utility of any other.

\begin{definition}\label{definition4}
    Let $\mathcal{G}$ and $s ', s \in \mathcal{S}$. Then a strategy profile $s '$ Pareto dominates strategy $s$ if, $u_i(s ') \geq u_i(s)$, for all $i$, and there exists some $j \in \mathcal{I}$ for which $u_j(s ') > u_j(s)$.
\end{definition}

Pareto domination is a useful notion to describe the social dilemma in a game. However, Pareto-dominated outcomes are often not played in Game Theory; a NE will always be preferred by rational players. For further discussion of the Game Theory notions presented above, see \cite{Shoham2008}.

Next, we provide our formulation and show that it is equivalent to the PD game.

\section{Game-Theoretical Formulation}\label{sec:game_formulation}

We consider a society of $n \in \mathbb{N}$, $n > 2$, travelers who seek to commute on a city's transportation network. We consider the road infrastructure as the common, yet limited, resource that is open-access and shared with all travelers. Each traveler has the option to utilize the roads by traveling in a CAV, which in turn contributes to the capacity of the roads. We expect each traveler to utilize the roads selfishly.

\begin{assumption}
    We assume full CAV-penetration, and so each traveler may choose either to travel in a CAV or use another mode of transportation, e.g., train, light rail, bicycling, or walking, thereby not contributing to congestion.
\end{assumption}

In a game-theoretic context, each traveler represents a rational player who has two possible actions, namely either $NC$ for not traveling in a CAV (non-CAV travel) or $C$ for traveling in a CAV (CAV travel). From now on, we shall use the terms ``player" and ``traveler," interchangeably.
All players receive a benefit $c \in \mathbb{R}_{> 0}$ for deciding to commute in the society. On the other hand, traveling using CAVs conveys benefits arising from flexibility, privacy, convenience, etc. So, if a player chooses to travel in a CAV, then they receive a benefit of $c + d$, where $d \in \mathbb{R}_{> 0}$ with $d \cdot (n - 2) > 2$ (this ensures that $d$ provides a significant incentive for CAV travel yet the lower bound decreases as $n$ increases). However, traveling in a CAV is naturally the selfish choice as it exploits the society's resources. Hence, for each player that decides to travel in a CAV, a cost of $e \in \mathbb{R}_{> 0}$ is imposed to the society as a whole and is paid out equally by all players, i.e., we define $\phi = e / n$ as the damage done to society. Without losing any theoretical insight, let us define $e = d + 1$ and assume that the original benefit $c$ is strictly greater than $e$.



\begin{remark}
    In our formulation, we want to capture the potential consequence of the players' decision to travel in a CAV. For this reason, contributing to the capacity of the roads (creating congestion, pollution, etc.) is represented by the cost and overall by the damage done to society.
\end{remark}

We can write the final form of player's $i$ payoff for traveling in a CAV as $(c + d) - (n - k) \phi$, and accordingly, player's $i$ payoff for not traveling in a CAV as $c - (n - k - 1) \phi$, where $k$ is the number of players who choose not to travel in a CAV other than player $i$. Thus, the payoff function is
\begin{equation}\label{absolute_payoffs}
    f_i(s_i, k) = 
        \begin{cases}
            c - (n - k - 1) \phi, & \; \text{if } s_i = NC, \\
            c + d - (n - k) \phi, & \; \text{if } s_i = C.
        \end{cases}
\end{equation}
For player $i$ the benefit of traveling in a CAV is denoted by $f_i(C, k)$ and the benefit of not traveling in a CAV by $f_i(NC, k)$, where $k$ is the number of players who decide to non-CAV travel other than player $i$. Note that \eqref{absolute_payoffs} depends not only on player $i$'s own action but also on $k$.

At this point, we can formally formulate our game denoted by $\mathcal{G}$. We have the finite set of players $\mathcal{I} = \{1, \ldots, n\}$ with $n > 2$; for each player $i$ the action set is $s_i \in \{NC, C\}$, and $f_i(s_i, k)$ with $k = 0, 1, \ldots, n - 1$ is the payoff function of player $i$. Thus, our game can be represented by the following tuple:
\begin{equation}\label{definition_of_game}
    \mathcal{G} = \left \langle \mathcal{I}, (S_i = \{NC, C\})_{i \in \mathcal{I}}, (f_i(s_i, k))_{i \in \mathcal{I}} \right \rangle.
\end{equation}

Next, we fully characterize game $\mathcal{G}$.

\begin{lemma}\label{payoff_difference}
    The payoff difference $\alpha = f_i(C, k) - f_i(NC, k)$ is positive and constant for all values $k \in [0, n - 1]$ and for all players $i \in \mathcal{I}$. Furthermore, $f_i(NC, k)$ and $f_i(C, k)$ are strictly increasing in $k$.
\end{lemma}



\begin{proof}
    We have $f_i(C, k) = c + d - (n - k) \phi$ and $f_i(NC, k) = c - (n - k - 1) \phi$ and so the difference is simply $f_i(C, k) - f_i(NC, k) = c + d - (n - k) \phi - [c - (n - k - 1) \phi] = d - \phi$. Hence, $\alpha$ is clearly positive by definition of $c$ and $d$ and also constant for all values $k = [0, n - 1]$. Furthermore, for $k > k '$, we have
        \begin{align}\label{a}
            f_i(NC, k) & = c - (n - k - 1) \phi, \quad \text{and} \\
            f_i(NC, k ') & = c - (n - k ' - 1) \phi. \label{b}
        \end{align}
    Subtracting \eqref{b} from \eqref{a} gives $f_i(NC, k) - f_i(NC, k ') = (k - k ') \phi > 0$,
    and so $f_i(NC, k) > f_i(NC, k ')$ for all $k$. In similar lines, we can show that the benefit of CAV travel, $f_i(C, k)$, is strictly increasing in $k$. Therefore, we conclude that $f_i(NC, k)$ and $f_i(C, k)$ are strictly increasing in $k$.
\end{proof}

From now on, the payoff difference is denoted by $\alpha$. We observe that the payoff difference, interpreted as the non-CAV travel cost, increases as $n$ increases. Interestingly enough, the payoff difference is independent of how many players choose not to travel in a CAV. In game-theoretic terms, we can interpret this as the strategy CAV travel dominating strategy non-CAV travel with a degree that is constant and independent of the other players who choose to CAV travel.

\begin{lemma}\label{nonnegativity and pareto relation}
    The payoff function \eqref{absolute_payoffs} is non-negative for all $k \in [0, n - 1]$, i.e., $f_i \geq 0$ for all $i \in \mathcal{I}$. Furthermore, mutual non-CAV travel is preferred to mutual CAV travel, i.e., $f_i(NC, n - 1) > f_i(C, 0)$ is a Pareto relation.
\end{lemma}

\begin{proof}
    We have $f_i(NC, 0) = c - (n - 1) \phi = c - (d + 1) + \phi$ and $f_i(C, 0) = c + d - n \cdot \phi = c - 1$. As $f_i(NC, k)$ and $f_i(C, k)$ are increasing in $k$, the result follows. Also, we have $f_i(C, 0) = c - 1$ and $f_i(NC, n - 1) = c - (n - (n - 1) - 1) \phi = c$ leading to $f_i(NC, n - 1) > f_i(C, 0)$ for all $i \in \mathcal{I}$.
\end{proof}

Lemma \ref{nonnegativity and pareto relation} establishes the fact that game $\mathcal{G}$ induces a Pareto relation, which implies that the equilibrium of mutual CAV travel is Pareto inferior to the alternative outcome, i.e., all players choose to non-CAV travel. This is significant since Pareto relations are directly associated with social dilemmas.

\begin{figure}[ht]
    \centering
        \includegraphics[width=1 \columnwidth]{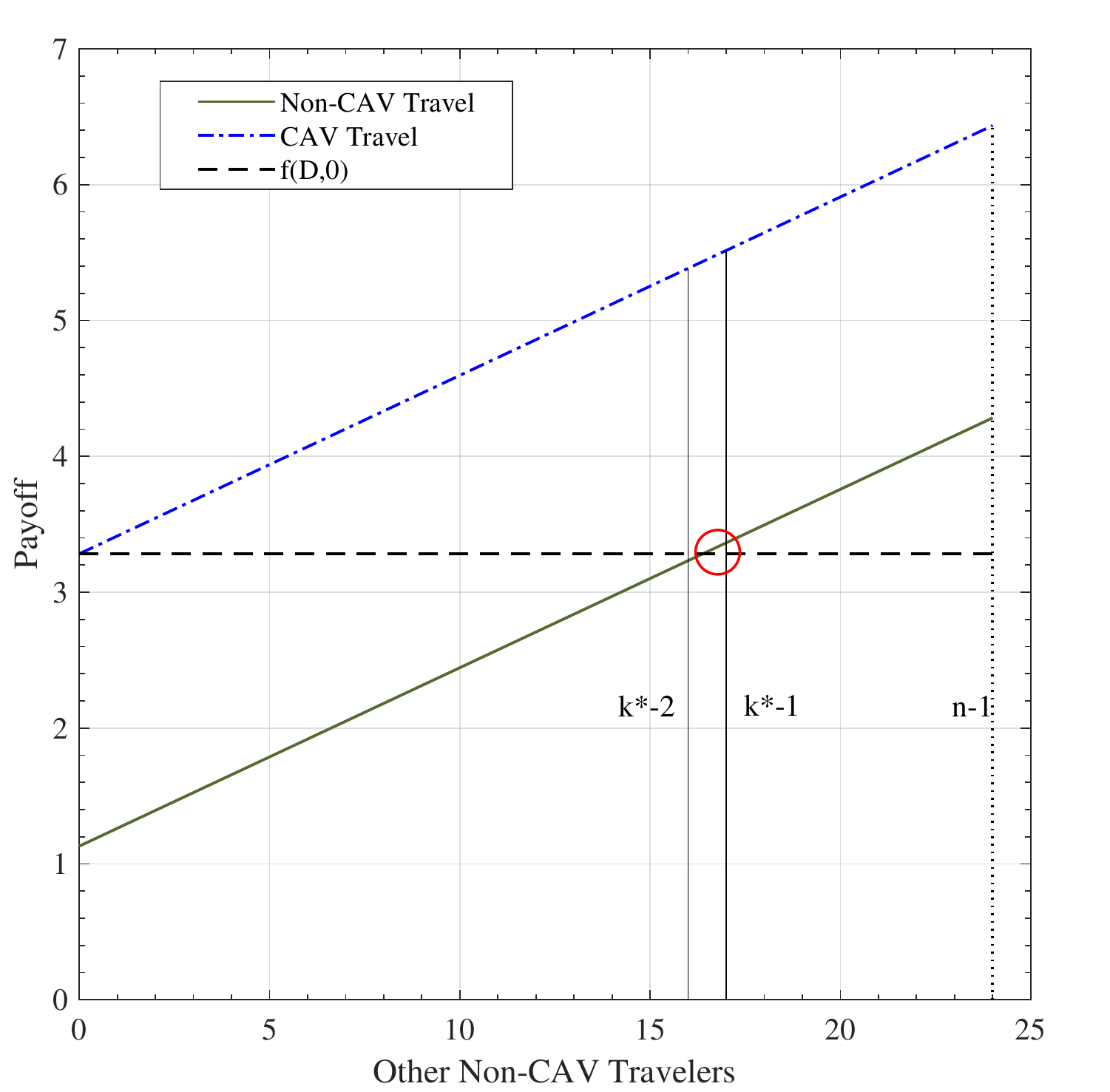}
            \caption{A visualization of the payoff function \eqref{absolute_payoffs} evaluated using the values $d = 2.2827$, $c = 4.2827$, and $n = 25$. We notice that, by focusing on the red circle, with a certain number of non-CAV travelers the overall utility of non-CAV travel is greater than the utility of CAV travel. This is the true meaning of a social dilemma in a CAV transportation context.}
    \label{payoff_visualization}
\end{figure}

\begin{theorem}\label{pr:PD}
    Game $\mathcal{G}$ defined in \eqref{definition_of_game} is equivalent to the PD game as both games share an equivalent incentive structure.
\end{theorem}

\begin{proof}
    By Lemma \ref{payoff_difference}, we have $f_i(NC, k) < f_i(C, k)$ for all $k \in [0, n - 1]$ which implies that the dominant strategy by rational players in the game is CAV travel no matter how many players decide to non-CAV travel. By Lemma \ref{nonnegativity and pareto relation}, the social dilemma induced structure is equivalent to that of the Prisoner's dilemma.
\end{proof}

\begin{corollary}
    The game defined in \eqref{definition_of_game} and the PD game provide equivalent incentives to the players, and thus, they result in equivalent outcomes.
\end{corollary}

Next, we show that by construction of the payoff function \eqref{absolute_payoffs}, non-CAV travel is more attractive from both the societal and the player's perspective.

\begin{proposition}
    Consider the game $\mathcal{G}$ defined in \eqref{definition_of_game}. Note that the benefit of CAV travel is given by $f_i(C, k)$ and the cost of non-CAV travel given by $\alpha$ (i.e., the payoff difference). Then the strategy non-CAV travel is socially desirable:
        \begin{equation}\label{socially_desirable}
            n f_i(C, k + 1) - (k + 1) \alpha > n f_i(C, k) - k \alpha, \; \forall i \in \mathcal{I}
        \end{equation}
    and also individually desirable:
        \begin{equation}\label{individually_desirable}
            f_i(NC, k + 1) > f_i(NC, k), \; \forall i \in \mathcal{I}.
        \end{equation}
\end{proposition}

\begin{proof}
    Both \eqref{socially_desirable} and \eqref{individually_desirable} can be verified by substitution of the corresponding functions in \eqref{absolute_payoffs}.
\end{proof}

Before we continue, let us introduce the notation $\lfloor x \rfloor$, which denotes the greatest integer that is less than $x$.

\begin{proposition}\label{minimum_number_cooperators}
    Consider game $\mathcal{G}$ defined in \eqref{definition_of_game}. There exists a unique integer $2 \leq k ^ * \leq n$ given by $k ^ * = \left \lfloor \frac{n d}{d + 1} \right \rfloor + 1$ such that
        \begin{equation}\label{k*_inequality}
            f_i(NC, k ^ * - 2) < f_i(C, 0) < f_i(NC, k ^ * - 1),
        \end{equation}
    where $k ^ *$ is the minimum number of non-CAV travelers.
\end{proposition}


\begin{proof}
    By substitution, we get the following equations:
        \begin{align}\nonumber
            f_i(NC, k ^ * - 2) & = c - (n - (k ^ * - 2) - 1) \phi \\
            & = c - (n - k ^ * + 1) \phi, \\
            f_i(C, 0) & = c - 1, \\ \nonumber
            f_i(NC, k ^ * - 1) & = c - (n - (k ^ * - 1) - 1) \phi \\
            & = c - (n - k ^ *) \phi.
        \end{align}
    We want to find a unique $k ^ *$ such that \eqref{k*_inequality} holds. So, we have
        \begin{equation}
            c - (n - k ^ * + 1) \phi < c - 1 < c - (n - k ^ *) \phi
        \end{equation}
    which leads to
        \begin{equation}\label{k*_derivation}
            \frac{n d}{d + 1} < k ^ * < \frac{n d}{d + 1} + 1.
        \end{equation}
    As $k ^ *$ is an integer, the last inequality \eqref{k*_derivation} is true if and only if $k ^ * = \left \lfloor \frac{n d}{d + 1} \right \rfloor + 1$ and $\frac{n d}{d + 1}$ is not an integer number.
\end{proof}


Proposition \ref{minimum_number_cooperators} intuitively implies that we need at least $k ^ *$ non-CAV travelers so that the benefit a player receives when they decide non-CAV travel will be greater than the dominant strategy $f(C, 0)$ (see in Fig. \ref{payoff_visualization} the red circle).

Next, we seek a way to characterize an outcome of the game in terms of preference. Now, in most cases, identifying the ``best" outcome is not possible, but there are certain situations that might be better from a societal standpoint.

\begin{proposition}\label{pareto_domination}
    The strategy of universal CAV travel, $f(C, 0)$, is Pareto dominated by outcomes with $k ' \geq k ^ * - 1$.
\end{proposition}


\begin{proof}
    We want to show that the outcomes with $k ' \geq k ^ * - 1$ Pareto dominate the dominant strategy of universal CAV travel. We only have to check two cases, namely $k ' \geq k ^ * - 1$ and $k ' < k ^ * - 1$. For $k ' = k ^ * - 1$, we have
        \begin{equation}
            f_i(NC, k ') = c - \left(n - \left \lfloor \frac{n d}{d + 1} \right \rfloor - 1 \right) \phi.
        \end{equation}
    Let $\left \lfloor \frac{n d}{d + 1} \right \rfloor = \frac{n d}{d + 1} - \varepsilon$, where $\varepsilon > 0$, so that
        \begin{align}\nonumber
            f_i(NC, k ') & = c - \left(n - \frac{n d}{d + 1} - \varepsilon - 1 \right) \left (\frac{d + 1}{n} \right) \\
            & = c - 1 + (\varepsilon + 1) \phi.
        \end{align}
    Subtracting $f_i(C, 0)$ from $f_i(NC, k ')$ gives $(\varepsilon + 1) \phi > 0$. Furthermore, for $k ' > k ^ * - 1$, note that $f_i$ is a strictly increasing function in $k$, thus $f_i(NC, k ') > f_i(NC, k ^ * - 1)$ which implies $f_i(NC, k ') > f_i(C, 0)$. Thus, for all players $i$, $f_i(NC, k ') > f_i(C, 0)$, where $k ' \geq k ^ * - 1$. On the other hand, if $k ' < k ^ * - 1$, then $f_i(NC, k') < f_i(C, 0)$ for all players $i$ by the first inequality relation in Proposition \ref{minimum_number_cooperators}. Hence, all outcomes which satisfy $k ' \geq k ^ * - 1$ Pareto dominate the dominant strategy of universal CAV travel.
\end{proof}


We note that by construction, the payoff function \eqref{absolute_payoffs} mutual non-CAV travel is the social optimum but, as a consequence of Proposition \ref{pareto_domination}, the decision to non-CAV travel is worthwhile to a player only if there are $k ^ *$ or more non-CAV travelers. Otherwise, everyone is no worse off at the dominant strategy of universal CAV travel. This gives rise to the notion of the state of \textit{minimally effective non-CAV travel}.

\begin{definition}
    The state of minimally effective non-CAV travel is the minimum number of non-CAV travelers, $k ^ *$, such that an outcome Pareto dominates the universal CAV travel equilibrium.
\end{definition}

Clearly, the state of minimally effective non-CAV travel is given by Propositions \ref{minimum_number_cooperators} and \ref{pareto_domination}. This is an important notion that can help in the derivation of the optimal utilization of CAVs in the emerging transportation systems.

Next, we discuss two solution approaches applied in our game $\mathcal{G}$. Our goal is to derive conditions that ensure a coalition of non-CAV travel, which are at least as large as the minimum state of non-CAV travel.

\section{Nash Equilibria and the Population Threshold}\label{sec:first_approach}

\subsection{Preference Structure}

Usually, in Game Theory, we assume that players are only interested in their own payoff. One of our goals is to study, in a more realistic setting, the players' social behavior, and so we impose to our game $\mathcal{G}$ a ``preference structure."

A preference structure allows us to model a particularly interesting scenario: the rational players are interested not only on their own payoff but also on the relative payoff share they receive, i.e., how their standing compares to that of others \cite{Michiardi2003}. The authors in \cite{Bolton2000} designed the ``equity, reciprocity, and competition (ERC)" model which is a simple model capable of handling a large population of players with an ``adjusted utility" function constructed on the premise that players are motivated by both their pecuniary payoff and their relative payoff standing. Notice that we changed our terminology of payoff function to adjusted utility function here. We do this to differentiate the difference between the absolute payoffs that players get from \eqref{absolute_payoffs} and the adjusted payoffs players will get in a preference structure. One of the reasons we use the ERC model is because it has been successful in explaining the behavior of selfish players in social experiments than other standard modeling techniques.


Now, we observe that players rarely play against the same other players, and so it is reasonable enough to analyze each game as one-shot. To further justify this, we only have to argue that it is highly unlikely to ``meet" other travelers in a major metropolitan city. Let the absolute payoff of player $i$ be given by $f_i$ from \eqref{absolute_payoffs}.
Each player $i$ seeks to maximize the expected utility of her motivation function $v_i = v_i(f_i, \sigma_i)$, where
\begin{equation}\label{eqn_relative_share}
    \sigma_i = \sigma_i(f_i, \gamma, n) =
        \begin{cases}
            f_i / \gamma, \quad & \text{if } \gamma > 0, \\
            1 / n, \quad & \text{if } \gamma = 0.
        \end{cases}
\end{equation}
Equation \eqref{eqn_relative_share} represents player $i$'s relative share of the payoff and $\gamma = \sum_{j = 1} ^ {n} f_j$ is the total pecuniary payout. We can think of the motivation function $v_i$ as the expected benefits that drive the players' behavior. We assume that $v_i$ is twice differentiable.

Next, we allow each player to be characterized by $a_i / b_i$ which is the ratio of weights that are attributed to the pecuniary and relative components of the motivation function. For example, strict relativism is represented by $a_i / b_i = 0$, so $\arg \max_{\sigma_i} v_i(\gamma \sigma_i, \sigma_i) = \pi = 1 / 2$, where $\pi_i(\gamma)$ is implicitly defined by $v_i(\gamma \pi_i, \pi_i) = v_i(0, 1 / n)$ for $\pi_i \leq 1 / n$. Strict narrow self-interest is the limiting case $a_i / b_i \to \infty$, so $\arg \max_{\sigma_i} v_i(\gamma \sigma_i, \sigma_i) = 1$ and $s \to 0$ \cite{Bolton2000}. Based on the above, the adjusted utility function then is given by:
\begin{equation}\label{adjusted_utility}
    u_i(f_i, \sigma_i) = a_i q(f_i) + b_i r(\sigma_i),
\end{equation}
where $q(\cdot)$ is strictly increasing, strictly concave, and differentiable; $r(\cdot)$ is differentiable, concave, and has its maximum at $\sigma_i = 1 / n$. Let us discuss a simple example from \cite{Bolton2000}.
\begin{example}
    We can explicitly define both $q$ and $r$ as:
        \begin{equation}\label{specific_preference}
            q(f_i) = f_i \quad \text{and} \quad r(\sigma_i) = - \frac{1}{2} \left (\sigma_i - \frac{1}{n} \right ) ^ 2,
        \end{equation}
    where function $q(\cdot)$ expresses the standard preferences for the payoff functions \eqref{absolute_payoffs}; function $r(\cdot)$ describes in a precise way the collective importance of equal division of the payoffs (this is also called the ``comparative effect.") Consequently, the further the allocation moves from player $i$ receiving an equal share, the higher the loss from the comparative effect.
\end{example}


\subsection{Analysis for the Nash Equilibria and the Threshold of Non-CAV Travel}

Our analysis in this subsection follows \cite{Michiardi2003}, but we apply it to our game $\mathcal{G}$ defined in \eqref{definition_of_game} along with the preference structure. Our goal is to study what influences strategic agents to non-CAV travel in our game $\mathcal{G}$.

We start our analysis by looking at the necessary and sufficient conditions for player $i$ to non-CAV travel, i.e.,
\begin{equation}
    u_i(f_i(NC, k + 1)) \geq u_i(f_i(NC, k)).
\end{equation}
Equivalently, we have from \cite{Michiardi2003} that $a_i / b_i \leq \delta(k)$, where
\begin{equation}\label{cooperation_delta_condition}
    \delta(k) = \frac{r \bigg(\frac{f_i(NC, k + 1)}{n f_i(C, k + 1) - (k + 1) \alpha} \bigg) - r \bigg(\frac{f_i(C, k)}{n f_i(C, k) - k \alpha} \bigg)}{q \big(f_i(C, k) \big) - q \big(f_i(NC, k + 1) \big)}.
\end{equation}
From \eqref{cooperation_delta_condition}, we can deduce that player $i$ will non-CAV travel if, and only if, there is overcompensation for the loss in absolute gain by moving closer to the average gain \cite{Michiardi2003}. Hence, we can state the general conditions of a NE:
\begin{align}\label{1st_condition_NE}
    & a_i / b_i \leq \delta(k - 1), \quad && \text{for $k$ players non-CAV travel}, \\ \label{2nd_condition_NE}
    & a_i / b_i \geq \delta(k), \quad && \text{for $n - k$ players CAV travel}.
\end{align}
We now have a better understanding of how the number of other non-CAV travelers, and its value can make non-CAV travel a rational strategy.


\begin{lemma}\label{lemma_delta_positive}
    For a given distribution of ERC-types, $\delta (k - 1) > 0$ is necessary but not sufficient to get a coalition size of $k$ where $n - k$ players free-ride. For a given payoff structure with $\delta (k - 1) > 0$, there exist ERC-types such that $k$ is an equilibrium coalition size.
\end{lemma}

\begin{proof}
    If $\delta (k - 1) < 0$, it is impossible for a coalition to form in the game of size $k$. On the other hand, if $a_i / b_i > \delta (k - 1)$ then condition \eqref{1st_condition_NE} cannot hold for any player. However, conditions \eqref{1st_condition_NE} and \eqref{2nd_condition_NE} imply that if $\delta(k - 1) > 0$, then there are types $\left(a_i / b_i \right)_{i \in \mathcal{I}}$ such that $k$ players non-CAV travel and $n - k$ players free-ride.
\end{proof}

\begin{proposition}\label{existence}
    By construction of the game $\mathcal{G}$ together with the ERC preference structure, there always exists a Nash equilibrium of universal CAV travel.
\end{proposition}

Proposition \ref{existence} follows directly from Lemma \ref{lemma_delta_positive}. We are though interested in finding a threshold of players that decide to non-CAV travel. The next proposition will help us do that.

\begin{proposition}\label{derive_k_cooperation}
    The necessary condition for an equilibrium of non-CAV travel $\delta (k - 1) > 0$ is equivalent to
        \begin{multline}\label{necessary_condition_delta}
            n \left[(k - 1) f_i(C, k) - k f_i(C, k - 1) \right] + \\
             + \left[n f_i(C, k - 1) - (k - 1) \alpha \right] \left[2 k - n \right] > 0.
        \end{multline}
\end{proposition}

\begin{proof}
    In order to obtain $\delta (k ^ * - 1) > 0$, it is necessary that by choosing the strategy CAV travel, a player further deviates from the equal share $1 / n$ than by choosing strategy non-CAV travel, i.e.,
        \begin{align}\nonumber
            \frac{f_i(C, k - 1)}{n f_i(C, k - 1) - (k - 1) \alpha} & - \frac{1}{n} > \\
            & \frac{1}{n} - \frac{f_i(NC, k)}{n f_i(C, k) - k \alpha}.
        \end{align}
    Rearranging and by eliminating the denominators, we get
        \begin{multline}\nonumber
            n f_i(C, k - 1)(n f_i(C, k) - k \alpha) + \\
            + n (f_i(C, k) - \alpha)(n f_i(C, k - 1) - (k - 1) \alpha)\\
            - 2 (n f_i(C, k - 1) - (k - 1) \alpha)(n f_i(C, k) - k \alpha) > 0,
        \end{multline}
    where we have used $\alpha = f_i(C, k) - f_i(NC, k)$.
    Substituting the payoff functions from \eqref{absolute_payoffs} and further simplification yield
        \begin{multline}
            n \left[(k - 1) f_i(C, k) - k f_i(C, k - 1) \right] + n (2 k - n) f_i(C, k - 1) \\
            - (k - 1) \alpha [2 k - n] > 0.\label{inequality_leading_to_coalition}
        \end{multline}
    Simplifying \eqref{inequality_leading_to_coalition} gives
        \begin{multline}
            n \left[(k - 1) f_i(C, k) - k f_i(C, k - 1) \right] + \\
            \left[n f_i(C, k - 1) - (k - 1) \alpha \right] \left[2 k - n \right] > 0.
        \end{multline}
    Therefore, the result follows immediately.
\end{proof}

We are now ready to prove the main result of the section.

\begin{theorem}\label{first_big_result}
    For any given vector of types, a rational player chooses to non-CAV travel when at least half of the players non-CAV travel.
\end{theorem}

\begin{proof}
    We only have to check on what conditions relation \eqref{necessary_condition_delta} is positive. By construction the payoff functions are non-negative, and thus $n f_i(C, k - 1) - (k - 1) \alpha > 0$, i.e.,
        \begin{equation}
            n f_i(C, k - 1) - (k - 1) \alpha = n (c - 1) + (k - 1)(1 + \phi)
        \end{equation}
    which is clearly positive for all values of $n, c$, and $k$. Hence, the second component of \eqref{necessary_condition_delta} is positive for $2k - n > 0$. Next, we look at the first component of \eqref{necessary_condition_delta}. By substituting the payoff function from \eqref{absolute_payoffs}, we get
        \begin{equation}
            (k - 1) f_i(C, k) - k f_i(C, k - 1) = 1 - c,
        \end{equation}
    which is negative for all values of $c$. We observe though that the second component is much bigger and dominates the first component as long as $2 k - n > 0$. Hence, relation \eqref{necessary_condition_delta} is positive and we have $\delta (k - 1) > 0$ for $2 k > n$. Therefore, for any given vector of types, if a player cooperates at the equilibrium, then at least half of the players cooperate.
\end{proof}

The interpretation of Theorem \ref{first_big_result} is that for any coalition to exist with size $k \geq 2$, a minimum of $n / 2$ players must join. We showed that given the specific payoff structure of our game $\mathcal{G}$ and along with the ERC preference structure, a coalition of players choosing strategy non-CAV travel could be formed provided that it is rather large. Thus, even if we impose a social mechanism that enforces strategy non-CAV travel in a society of travelers and satisfying \eqref{necessary_condition_delta}, a coalition of at least size $n / 2$ must be formed to create an equilibrium of non-CAV travel. Therefore, the social mechanism will require significant influence over the players' behaviors in order to create a state of effective non-CAV travel. On the other hand, this result is promising as it shows that a social solution can potentially prevent self-centered and destructive behavior towards society.

\section{Creating an Institutional Arrangement}\label{sec:second_approach}


In this section, we take advantage of the equivalency of our game $\mathcal{G}$ to the Prisoner's dilemma in order to use the non-cooperative game model of institutional arrangements framework of \cite{Okada1993}. We prove in Theorem 3 that the ratio of non-CAV travel and CAV travel in a deregulated society as the number of players increases, tends to zero. In other words, as the society becomes larger and larger, the incentive to cooperatively agree not to travel in a CAV tends to zero.

Players are free to create a social institution that binds them by selecting their actions. In other words, players agree to have an institutional arrangement with the purpose of enforcing an agreement of non-CAV travel. The first stage is the creation of a social institution, and this is done through participation negotiations, and thus the first stage is called ``participation decision stage." All players have to decide whether they will participate in negotiations for collective decision making, or not, without any knowledge of each others' decisions. The outcome of the game at this first stage is either that some group of players is formed or not. All players decide to participate in negotiations or not based on their expectations about what will happen in the rest of the game. The possibility of non-CAV travel is significantly affected by the number of players. That means that the outcome of the institutional arrangements depends on the players' decisions in the first stage \cite{Okada1993}.

\begin{remark}
    In contrast to Cooperative Game Theory, there is no external binding enforcement, and players are free to make their decisions (whether it is beneficial to them only). Thus, we treat the institutional arrangements framework as a non-cooperative game.
\end{remark}

The goal here is to investigate the question: \textit{does the number of travelers affect the possibility of non-CAV travel?}

The next proposition addresses the basic cases.

\begin{proposition}(\hspace{1sp}\cite{Okada1993})
    Let $d_i = 1$ denote a player $i$'s decision to participate in bargaining for installing an enforcement agency; otherwise $d_i = 0$. When $k ^ * = n$, the participation decision stage has a unique solution $d ^ * = (1, \ldots, 1)$.
\end{proposition}

It is interesting enough that in the special case $n = 2$, both players agree to create an enforcement agency and also to non-CAV travel in the institutional arrangements.

\begin{definition}
    The incentive ratio of non-CAV travel and CAV travel can be defined as a positive number given by:
        \begin{equation}
            \beta = \frac{f_i(NC, k ^ * - 1) - f_i(C, 0)}{f_i(C, k ^ * - 1) - f_i(NC, k ^ * - 1)}.
        \end{equation}
\end{definition}

In words, $\beta$ represents the ratio of players' incentive to form the minimum group for non-CAV travel, i.e., the group of $k ^ *$ non-CAV travelers, over their incentive to deviate unilaterally from the minimum group for non-CAV travel. Given our game $\mathcal{G}$ defined in $\eqref{definition_of_game}$, we have
\begin{equation}\label{beta}
    \beta = \frac{k ^ * \phi - d}{d - \phi} = \frac{k ^ * (d + 1) - n d}{d (n - 1) - 1}.
\end{equation}

\begin{proposition}(\hspace{1sp}\cite{Okada1993})\label{prop:okada}
    The uncooperative solution of the institutional arrangements for our game $\mathcal{G}$ prescribes the following player behavior:
        \begin{enumerate}
            \item If $n = \left\lfloor \frac{n d}{d + 1} \right\rfloor + 1 = k ^ *$, then all players participate in bargaining and they agree to non-CAV travel.
            \item If $n \geq \left\lfloor \frac{n d}{d + 1} \right\rfloor + 2$, then every player participates in bargaining with probability $t(n)$ satisfying:
                \begin{equation}\nonumber
                    \beta = \sum_{k ^ * \leq k \leq n - 1}\frac{(n - k ^ *) \cdot \ldots \cdot(n - k)}{k ^ * \cdot \ldots \cdot k} \bigg( \frac{t}{1 - t} \bigg) ^ {k - k ^ * + 1},
                \end{equation}
        \end{enumerate}
    where $k ^ *$ and $\beta$ are given by Proposition \ref{minimum_number_cooperators} and \eqref{beta}, respectively.
\end{proposition}


We are ready now to prove our main result of this section, which has to do with the limiting behavior of $\beta$.

\begin{theorem}\label{limiting_behavior_beta}
    As the number of players increases, the incentive ratio of non-CAV travel and CAV travel vanishes, i.e., $\beta$ tends to zero as $n$ tends to infinity.
\end{theorem}

\begin{proof}
    Substitute $k ^ * = \left \lfloor \frac{n d}{d + 1} \right \rfloor + 1$ into $\beta$ to get
        \begin{equation}
            \beta = \frac{\left(\left \lfloor \frac{n d}{d + 1} \right \rfloor + 1 \right)(d + 1) - n d}{d (n - 1) - 1}.
        \end{equation}
%
%
    By Proposition \ref{minimum_number_cooperators}, $\frac{n d}{d + 1}$ is not an integer, thus we can write $\left \lfloor \frac{n d}{d + 1} \right \rfloor = \frac{n d}{d + 1} - \varepsilon$, where $\varepsilon > 0$. Now taking the limit of $\beta$ as $n$ goes to infinity gives
        \begin{equation}
            \lim_{n \to \infty} \beta = \lim_{n \to \infty} \frac{\left(\left \lfloor \frac{n d}{d + 1} \right \rfloor + 1 \right)(d + 1) - n d}{d (n - 1) - 1},
        \end{equation}
    or equivalently
        \begin{align}
            \lim_{n \to \infty} \beta & = \lim_{n \to \infty}\frac{(\frac{n d}{d + 1} - \varepsilon + 1)(d + 1) - n d}{d (n - 1) - 1} \\
            & = \lim_{n \to \infty} \frac{n d + ( - \varepsilon + 1)(d + 1) - n d}{d (n - 1) - 1} \\
            & = \lim_{n \to \infty} \frac{( - \varepsilon + 1)(d + 1)}{d (n - 1) - 1}.
        \end{align}
    We divide both numerator and denominator by $1 / n$ and using the standard limit $\lim_{x \to \infty} \frac{1}{x} = 0$ gives the result, i.e.,
        \begin{equation}
            \frac{\lim_{n \to \infty} \frac{( - \varepsilon + 1)(d + 1)}{n}}{d - \lim_{n \to \infty} \frac{d}{n} - \lim_{n \to \infty} \frac{1}{n}} = 0.
        \end{equation}
    Thus, we conclude that $\lim_{n \to \infty} \beta = 0$.
\end{proof}

To complement our understanding, we performed a numerical study of the limiting behavior of $t(n)$, given in Table \ref{tab:numerical}. In the table, we have included the additional probabilities: $p_A(n)$ shows the probability of some group of size $k ^ *$ or greater reaching an agreement, $p_I(n)$ the probability of each player being an insider of some group with at least $k ^ *$ non-CAV travelers, and $p_F(n)$ is the probability of each player being a free rider, i.e., existing outside of a group of at least $k ^ *$ non-CAV travelers.

\begin{table}[ht]
    \centering
    \begin{tabular}{c|c|c|c|c|c|c}
        $n$ & $k ^ *$ & $\beta$ & $t(n)$ & $p_A(n)$ & $p_I(n)$ & $p_F(n)$ \\ \midrule
        3 & 3 & 0.930 & 1.000 & 1.000 & 1.000 & 0.000 \\ 
        4 & 3 & 0.166 & 0.333 & 0.111 & 0.086 & 0.025 \\ 
        5 & 4 & 0.253 & 0.503 & 0.192 & 0.160 & 0.032 \\ 
        6 & 5 & 0.302 & 0.602 & 0.236 & 0.204 & 0.031 \\ 
        7 & 5 & 0.066 & 0.139 & 0.001 & 0.001 & 0.000 \\ 
        8 & 6 & 0.129 & 0.269 & 0.006 & 0.005 & 0.001 \\ 
        9 & 7 & 0.175 & 0.363 & 0.014 & 0.011 & 0.003 \\ 
        10 & 7 & 0.037 & 0.078 & 0.000 & 0.000 & 0.000 \\ 
        11 & 8 & 0.083 & 0.174 & 0.000 & 0.000 & 0.000 \\ 
        12 & 9 & 0.120 & 0.252 & 0.000 & 0.000 & 0.000 \\ 
        13 & 9 & 0.023 & 0.048 & 0.000 & 0.000 & 0.000 \\ 
        14 & 10 & 0.059 & 0.124 & 0.000 & 0.000 & 0.000 \\ 
        15 & 11 & 0.089 & 0.188 & 0.000 & 0.000 & 0.000 \\ 
        20 & 14 & 0.034 & 0.072 & 0.000 & 0.000 & 0.000 \\ 
        25 & 17 & 0.002 & 0.004 & 0.000 & 0.000 & 0.000 \\ 
        30 & 21 & 0.033 & 0.070 & 0.000 & 0.000 & 0.000 \\ 
        35 & 24 & 0.011 & 0.023 & 0.000 & 0.000 & 0.000 \\ 
        40 & 28 & 0.033 & 0.069 & 0.000 & 0.000 & 0.000 \\ 
        45 & 31 & 0.016 & 0.033 & 0.000 & 0.000 & 0.000 \\ 
        50 & 34 & 0.002 & 0.004 & 0.000 & 0.000 & 0.000 
    \end{tabular}
    \caption{Numerical study for game $\mathcal{G}$ with the institutional arrangements where $d \approx 2$.}
    \label{tab:numerical}
\end{table}

\begin{figure}[ht]
    \centering
        \includegraphics[width=1 \columnwidth]{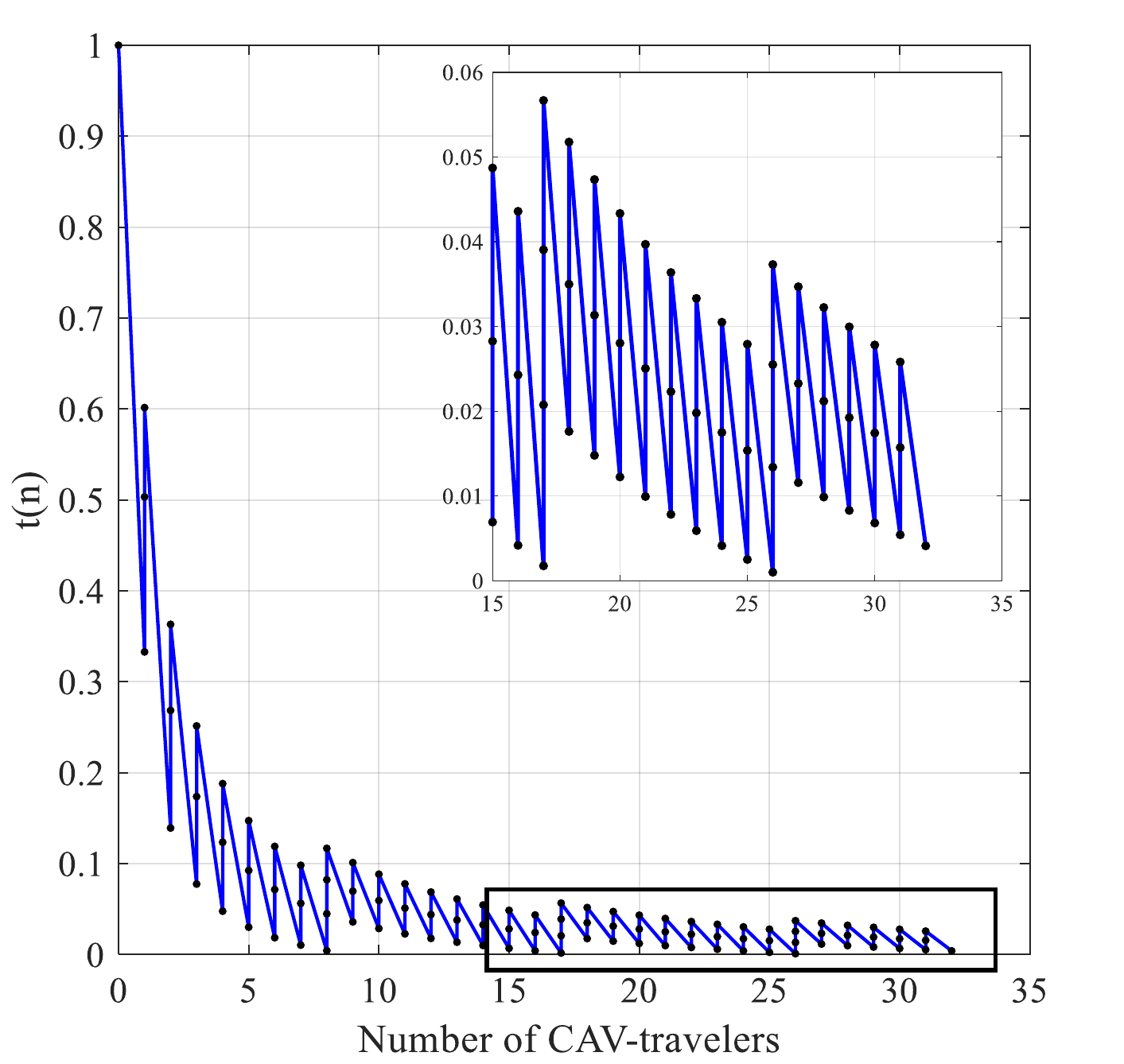}
        \caption{Plot of $t(n)$ as a function of the number of CAV travelers. The blue line shows the sequence of $t(n)$ as $n$ increases from $0$ to $100$.}\label{fig:tvsd}
\end{figure}

From Theorem \ref{limiting_behavior_beta}, the incentive ratio goes to zero as the number of players increases. In addition, from the numerical study summarized in Table \ref{tab:numerical} and Figure \ref{fig:tvsd}, the likelihood of bargaining for an institution, $t(n)$ and probability of being an insider, $p_A(n)$ approach zero as $n$ gets large. This implies that for large societies, the impact of self-realized non-CAV travel is non-existent, and the universal CAV travel strategy dominates. For small societies with $k ^ * = n$, it is a certainty that players agree to bargain and create an institution for CAV travel (which is not ideal).

\section{Conclusion}\label{sec:conclusions}

In this paper, we addressed the problem of the social consequences of decision-making of human interaction with connectivity and automation in a game-theoretic setting. We formulated the problem as a multiplayer normal-form game and showed that the incentive structure is equivalent to the PD game. The proposed approach has the benefit of capturing the social dilemma that is expected to arise from the future social-mobility dilemma. We considered two different approaches: one was with a preference structure and one with institutions. We investigated and derived conditions for the unselfish strategy, i.e., non-CAV travel, to appear in the game. In the first case, we came up with conditions for a NE and derived a threshold for non-CAV travel; in the second case, we allowed players to create an institution that can enforce non-CAV travel. We concluded that the incentive ratio of non-CAV travel over CAV travel tends to zero as the number of players increases.

Ongoing work includes the design of a framework that analyzes the impact of decision-making by relaxing the assumptions of complete information (e.g., inducing a Bayesian setting) aiming to capture the informational limitations of players in the game.

\bibliographystyle{IEEEtran}
\bibliography{references}

\begin{thebibliography}{10}
\providecommand{\url}[1]{#1}
\csname url@samestyle\endcsname
\providecommand{\newblock}{\relax}
\providecommand{\bibinfo}[2]{#2}
\providecommand{\BIBentrySTDinterwordspacing}{\spaceskip=0pt\relax}
\providecommand{\BIBentryALTinterwordstretchfactor}{4}
\providecommand{\BIBentryALTinterwordspacing}{\spaceskip=\fontdimen2\font plus
\BIBentryALTinterwordstretchfactor\fontdimen3\font minus
  \fontdimen4\font\relax}
\providecommand{\BIBforeignlanguage}[2]{{%
\expandafter\ifx\csname l@#1\endcsname\relax
\typeout{** WARNING: IEEEtran.bst: No hyphenation pattern has been}%
\typeout{** loaded for the language `#1'. Using the pattern for}%
\typeout{** the default language instead.}%
\else
\language=\csname l@#1\endcsname
\fi
#2}}
\providecommand{\BIBdecl}{\relax}
\BIBdecl

\bibitem{Marletto2019}
G.~Marletto, ``Who will drive the transition to self-driving? {A}
  socio-technical analysis of the future impact of automated vehicles,''
  \emph{Technological Forecasting and Social Change}, vol. 139, pp. 221--234,
  2019.

\bibitem{Bernard2014}
A.~Bernard, \emph{Lifted: A Cultural History of the Elevator}.\hskip 1em plus
  0.5em minus 0.4em\relax NYU Press, 2014.

\bibitem{Zmud2017}
J.~P. Zmud and N.~S. Ipek, ``Towards an understanding of the travel behavior
  impact of autonomous vehicles,'' \emph{Transportation research procedia 25},
  pp. 2500--2519, 2017.

\bibitem{Taiebat2018}
M.~Taiebat, A.~L. Brown, H.~R. Safford, S.~Qu, and M.~Xu, ``A review on energy,
  environmental, and sustainability implications of connected and automated
  vehicles,'' \emph{Environmental Science \& Technology}, vol. 52.20, pp.
  11\,449--11\,465, 2018.

\bibitem{Dawes1973}
R.~M. Dawes, ``The commons dilemma game: An n-person mixed-motive game with a
  dominating strategy for defection,'' \emph{{ORI} Research Bulletin}, vol.
  13.2, pp. 1--12, 1973.

\bibitem{Dawes1980}
M.~Dawes, ``Social dilemmas,'' \emph{Annual review of psychology}, vol. 31.1,
  pp. 169--193, 1980.

\bibitem{Platt1973}
J.~Platt, ``Social traps,'' \emph{American Psychologist}, vol.~28, pp.
  641--651, 1973.

\bibitem{Stern1992}
P.~C. Stern, ``Psychological dimensions of global environmental change,''
  \emph{Annual Review of Psychology}, vol.~43, pp. 269--302, 1992.

\bibitem{Joshi2005}
M.~S. Joshi, V.~Joshi, and R.~Lamb, ``The prisoners' dilemma and city-centre
  traffic,'' \emph{Oxford Economic Papers}, vol. 57.1, pp. 70--89, 2005.

\bibitem{Whitelegg1997}
J.~Whitelegg, \emph{Critical Mass: Transport, Environment and Society in the
  Twenty-first Century}.\hskip 1em plus 0.5em minus 0.4em\relax Pluto Press,
  1997.

\bibitem{Okada1993}
A.~Okada, ``The possibility of cooperation in an n-person prisoners' dilemma
  with institutional arrangements,'' \emph{Public Choice}, vol. 77.3, pp.
  629--656, 1993.

\bibitem{Mehr2018}
N.~Mehr and R.~Horowitz, ``Can the presence of autonomous vehicles worsen the
  equilibrium state of traffic networks?'' \emph{IEEE Conference on Decision
  and Control ({CDC})}, 2018.

\bibitem{Jackeline2016a}
J.~Rios-Torres and A.~A. Malikopoulos, ``Energy impact of different
  penetrations of connected and automated vehicles: a preliminary assessment,''
  \emph{Proceedings of the 9th ACM SIGSPATIAL International Workshop on
  Computational Transportation Science}, 2016.

\bibitem{Jackeline2016b}
------, ``An overview of driver feedback systems for efficiency and safety,''
  \emph{2016 IEEE 19th International Conference on Intelligent Transportation
  Systems (ITSC)}, 2016.

\bibitem{Malikopoulos2018}
A.~A. Malikopoulos, C.~G. Cassandras, and Y.~J. Zhang, ``A decentralized
  energy-optimal control framework for connected automated vehicles at
  signal-free intersections,'' \emph{Automatica}, vol.~93, pp. 244--256, 2018.

\bibitem{Jackeline2017}
J.~Rios-Torres and A.~A. Malikopoulos, ``A survey on the coordination of
  connected and automated vehicles at intersections and merging at highway
  on-ramps,'' \emph{IEEE Transactions on Intelligent Transportation Systems},
  vol. 18.5, pp. 1066--1077, 2017.

\bibitem{Singleton2018}
P.~A. Singleton, ``Discussing the ``positive utilities" of autonomous vehicles:
  will travellers really use their time productively?'' \emph{Transport
  Reviews}, vol. 39.1, pp. 1--16, 2018.

\bibitem{Soteropoulos2019}
A.~Soteropoulos, M.~Berger, and F.~Ciari, ``Impacts of automated vehicles on
  travel behaviour and land use: an international review of modelling
  studies,'' \emph{Transport Reviews}, vol. 39.1, pp. 29--49, 2019.

\bibitem{Auld2017}
J.~Auld, V.~Sokolov, and T.~S. Stephens, ``Analysis of the effects of
  connected-automated vehicle technologies on travel demand,'' \emph{Journal of
  the Transportation Research Board (2625)}, pp. 1--8, 2017.

\bibitem{Sarkar2016}
R.~Sarkar and J.~Ward, ``{DOE} smart mobility: Systems and modeling for
  accelerated research in transportation,'' \emph{Road Vehicle Automation},
  vol.~3, pp. 39--52, 2016.

\bibitem{Zhao2019}
L.~Zhao and A.~A. Malikopoulos, ``Enhanced mobility with connectivity and
  automation: A review of shared autonomous vehicle systems,'' \emph{IEEE
  Intelligent Transportation Systems Magazine}, 2020.

\bibitem{Mehr2019a}
N.~Mehr and R.~Horowitz, ``Pricing traffic networks with mixed vehicle
  autonomy,'' \emph{IEEE American Control Conference (ACC)}, 2019.

\bibitem{Shoham2008}
Y.~Shoham and K.~Leyton-Brown, \emph{Multiagent Systems: Algorithmic,
  Game-Theoretic, and Logical Foundations}.\hskip 1em plus 0.5em minus
  0.4em\relax Cambridge University Press, 2008.

\bibitem{Michiardi2003}
P.~Michiardi and R.~Molva, ``Ad hoc networks security,'' \emph{ST Journal of
  System Research}, 2003.

\bibitem{Bolton2000}
G.~E. Bolton and A.~Ockenfels, ``{ERC}: a theory of equity, reciprocity, and
  competition,'' \emph{The American Economic Review}, vol.~90, pp. 166--193,
  2000.

\end{thebibliography}

\end{document}